\documentclass[A4paper]{article}
\usepackage{amsmath}
\usepackage{amssymb}
\usepackage{amsthm}
\usepackage{graphicx,psfrag,epsf}
\usepackage{enumerate}
\usepackage{natbib}
\usepackage{url} 
\usepackage{fullpage}

\usepackage{pdflscape}
\usepackage{enumitem}
\usepackage{algpseudocode}

\newcommand{\bc}{\mathbf{c}}

\newcommand{\bx}{\mathbf{x}}

\newcommand{\br}{\mathbf{r}}

\newcommand{\bt}{\mathbf{t}}
\newcommand{\by}{\mathbf{y}}
\newcommand{\bff}{\mathbf{f}}

\newcommand{\bb}{\mathbf{b}}
\newcommand{\bu}{\mathbf{u}}

\newcommand{\bz}{\mathbf{z}}

\newcommand{\balpha}{\boldsymbol{\alpha}}
\newcommand{\bbeta}{\boldsymbol{\beta}}

\newcommand{\bmu}{\boldsymbol{\mu}}
\newcommand{\bchi}{\boldsymbol{\chi}}
\newcommand{\btheta}{\boldsymbol{\theta}}

\newcommand{\btau}{\boldsymbol{\tau}}
\newcommand{\brho}{\boldsymbol{\rho}}

\newcommand{\bzero}{\mathbf{0}}

\def\T{{ \mathrm{\scriptscriptstyle T} }}

\usepackage{amsthm}

\newtheorem{theorem}{Theorem}[section]

\newtheorem{lemma}[theorem]{Lemma}

\begin{document}

\def\spacingset#1{\renewcommand{\baselinestretch}%
{#1}\small\normalsize} \spacingset{1}

\def\T{{ \mathrm{\scriptscriptstyle T} }}

  \title{\bf Gibbs optimal design of experiments}
  \author{
    Antony M. Overstall \\
    School of Mathematical Sciences, University of Southampton,\\ Southampton SO17 1BJ, United Kingdom,\\ A.M.Overstall@soton.ac.uk\\
    and \\
    Jacinta Holloway-Brown \\
    School of Computer and Mathematical Sciences, University of Adelaide,\\
    Adelaide 5005, Australia,\\
    jacinta.holloway-brown@adelaide.edu.au\\
    and \\
    James M. McGree \\
    School of Mathematical Sciences, Queensland University of Technology,\\ Brisbane 4001, Australia,\\ james.mcgree@qut.edu.au
    }
  \maketitle

\bigskip
\begin{abstract}
Bayesian optimal design is a well-established approach to planning experiments. A distribution for the responses, i.e. a statistical model, is assumed which is dependent on unknown parameters. A utility function is then specified giving gain in information in estimating the true values of the parameters, using the Bayesian posterior distribution. A Bayesian optimal design is given by maximising expectation of the utility with respect to the distribution implied by statistical model and prior distribution for the true parameter values. The approach accounts for the experimental aim, via specification of the utility, and of assumed sources of uncertainty. However, it is predicated on the statistical model being correct. Recently, a new type of statistical inference, known as Gibbs inference, has been proposed. This is Bayesian-like, i.e. uncertainty for unknown quantities is represented by a posterior distribution, but does not necessarily require specification of a statistical model. The resulting inference is less sensitive to misspecification of the statistical model. This paper introduces Gibbs optimal design: a framework for optimal design of experiments under Gibbs inference. A computational approach to find designs in practice is outlined and the framework is demonstrated on exemplars including linear models, and experiments with count and time-to-event responses.
\end{abstract}

\noindent%
\vfill

\newpage
\spacingset{1} 

\section{Introduction} \label{sec:intro}

Experiments are key to the scientific method and are a pillar of statistical science \citep[e.g.][Chapter 6]{stigler_2016}. They are used to systematically investigate the relationship between a series of controllable variables and a response variable. In this paper, an experiment consists of a fixed number of runs, where each run involves the specification of all controllable variables and the subsequent observation of a response. The experimental aim is then addressed, usually by the estimation of a statistical model, or models. The quality of this analysis can be highly dependent on the design of the experiment: the specification of the controllable variables for each run \citep[see, e.g., ][]{mead_etal_2012}. 

\emph{Decision-theoretic Bayesian optimal design of experiments} \citep[][termed here as \emph{Bayesian optimal design} for brevity]{chalonerverdinelli1995} provides a principled approach to planning experiments. The Bayesian optimal design approach can briefly be described as follows. First, a statistical model is assumed. This is a probability distribution for the responses, fully specified apart from a vector of unknown parameters. A prior distribution is also assumed representing prior knowledge about the true values of the parameters, i.e. those parameter values that make the statistical model coincide with the true probability distribution for the responses. A utility function is then specified, returning the gain in information, in estimating the true values of the parameters, from the responses obtained via a given design. The expected utility is given by the expectation of the utility with respect to the responses and true parameter values, under the joint probability distribution implied by the statistical model and prior distribution. Lastly, a Bayesian optimal design maximises the expected utility over the space of all designs.

The advantages of the Bayesian optimal design approach include the following. Firstly, the approach explicitly takes account of the aim of the experiment through specification of the utility function. For example, the approach can easily be extended for model selection and/or prediction. Second, by taking the expectation of the utility, it incorporates all known sources of uncertainty. 

The patent disadvantage is that the objective function given by the expected utility is rarely available in closed form. Instead, an often computationally expensive numerical routine is used to approximate the expected utility which then needs to be maximised over a potentially high-dimensional design space. However, in the last decade, new computational methodology has been proposed to find Bayesian designs more efficiently than previously possible, or find designs for scenarios which were previously unattainable \citep[see, e.g., ][and reference therein]{rainforth_etal_2023}.

A more fundamental disadvantage is that the statistical model for the responses is specified before the responses are observed. The resulting design can be highly tailored to a statistical model which may actually be significantly misspecified, compromising the accuracy and precision of future statistical inference. 

Recently, there has been progress in Bayesian-like statistical inference that does not require the specification of a probability distribution, i.e. a statistical model. Instead a loss function is specified identifying desirable parameter values for given responses (those that minimise the loss). The specification of the loss does not necessarily follow from specification of a probability distribution for the responses. A so-called (unnormalised) Gibbs posterior distribution for the parameters is then given by the product of the exponential of the negative loss and a prior distribution for target parameter values. The target parameter values are the parameters which minimise the expectation of the loss with respect the true (but unknown and unspecified) probability distribution of the responses \citep{bissiri_etal_2016}. Gibbs inference (inference using the Gibbs posterior distribution) can be seen as a Bayesian-like analogue of classical M-estimation \citep[e.g., ][Chapter 7]{hayashi_2000}. \cite{bissiri_etal_2016} provide a thorough theoretical treatment of Gibbs inference. In particular, they show that the Gibbs posterior distribution provides coherent inference about the target parameter values. 


This paper proposes a \emph{decision-theoretic Gibbs optimal design of experiments} framework (referred to as \emph{Gibbs optimal design} for brevity). Similar to Bayesian optimal design, a utility function is specified: a function returning gain in information, in estimating the target parameter values, using responses obtained via a given design, where dependence on the responses is through the Gibbs posterior distribution. However the expected utility does not immediately follow. The absence of a statistical model, means there is no joint probability distribution for the responses and target parameter values. Instead, we propose taking expectation of the utility function with respect to a probability distribution for the responses which we term the designer distribution. This user-specified designer distribution should be flexible enough to be close to the true probability distribution for the responses, certainly closer than the statistical model. Specification of the designer distribution can be aided by the fact that it does not need to be ``useful", i.e. it will not be fitted to the observed responses of the experiment and need not be capable of addressing the experimental aim.

To demonstrate the concept of the designer distribution, consider the following simple example. Suppose the aim of the experiment is to investigate the relationship between a single controllable variable, $x$, and a continuous response. It is assumed that the mean response is given by a polynomial of $x$ where the coefficients are unknown parameters. A reasonable loss is sum of squares: the sum of the squared differences between the responses and their mean (given by the polynomial). A suitable designer distribution is the unique-treatment model \citep[e.g.][]{gilmourtrinca2012}: a probability distribution where each unique $x$ value in the experiment exhibits a unique mean. The unique-treatment model is not a useful model, i.e. it does not allow one to effectively learn the relationship between $x$ and the response. However, it is a flexible model likely to be closer to capturing the true relationship between $x$ and the response, than a polynomial. 

The topic considered in this paper fits within the field of \emph{robust design of experiments}. One of the first treatments of this topic was the work of \cite{box_draper_1959} who considered design of experiments for a statistical model with mean response a linear model with an intercept and first-order term, for a single controllable variable, when the true probability distribution has a mean response with an additional quadratic term. A relatively up-to-date review of robust design of experiments is given by \cite{wiens_2015} who considers the topic under misspecification of both the mean and error structure of the statistical model, as well as from the point of view of different applications (including dose-response, clinical trials and computer experiments).

The remainder of the paper is organised as follows. Section~\ref{sec:back} provides brief mathematical descriptions of Bayesian optimal design and Gibbs inference, before Section~\ref{sec:god} introduces the concept of Gibbs optimal design. Section~\ref{sec:LM} considers Gibbs optimal design for the class of linear models. Finally, we demonstrate Gibbs optimal design on illustrative examples in Section~\ref{sec:examples}.

\section{Background} \label{sec:back}

\subsection{Design problem setup} \label{sec:setup}

Suppose the experiment aims to investigate the relationship between $k$ controllable variables denoted $\bx = \left(x_1, \dots, x_k \right)^\T   \in \mathbb{X}$, and a measurable response denoted by $y$. The experiment consists of $n$ runs, where the $i$th run involves specifying $\bx_i =  \left(x_{i1}, \dots, x_{ik} \right)^\T   \in \mathbb{X}$ and observing response $y_i$, for $i=1,\dots,n$. Let $\by = \left(y_1, \dots, y_n \right)^\T $ denote the $n \times 1$ vector of responses and $X$ the $n \times k$ design matrix with $i$th row $\bx_i^\T $, for $i=1,\dots,n$. It is assumed that $\by$ are realisations from a multivariate probability distribution denoted $\mathcal{T}(X)$, which is the true but unknown response-generating probability distribution.  This paper addresses the specification of $X$ to learn the most about the true response-generating probability distribution.

\subsection{Bayesian optimal design} \label{sec:bod}

Bayesian optimal design relies on the specification of a statistical model and utility function. The statistical model, denoted $\mathcal{S}(\bt; X)$, is a probability distribution for the responses $\by$ used to represent the true response-generating probability distribution, $\mathcal{T}(X)$. The statistical model is fully specified up to a $p \times 1$ vector of parameters $\bt = \left(t_1,\dots,t_p \right)^\T  \in \Theta$ with parameter space $\Theta \subset \mathbb{R}^p$.

Suppose there exist values $\btheta_T$ of the parameters such that the statistical model and true response-generating probability distribution coincide, i.e. $\mathcal{T}(X) = \mathcal{S}(\btheta_T; X)$. If $\btheta_T$ is known, then the true response-generating distribution is completely known so, in some form, estimating $\btheta_T$ is the aim of the experiment. Bayesian inference for this aim is achieved by evaluating the Bayesian posterior distribution, denoted $\mathcal{B}(\by; X)$ and given by
\begin{equation}
\pi_\mathcal{B}(\bt \vert \by; X) \propto \pi(\by \vert \bt ; X) \pi_T(\bt).
\label{eqn:bayes}
\end{equation}
In (\ref{eqn:bayes}), $\pi(\by \vert \bt ; X)$ is the likelihood function: the probability density/mass function of $\mathcal{S}(\bt; X)$, and $\pi_T(\cdot)$ is the probability density function (pdf) for the prior distribution of $\btheta_T$ (denoted $\mathcal{P}_T$). For simplicity, we assume that the prior distribution $\mathcal{P}_T$ does not depend on the design $X$. However, the technicalities of Bayesian optimal design change little if $\mathcal{P}_T$ depends on $X$.

The utility function is denoted $u_{\mathcal{B}}(\bt, \by, X)$. It gives the gain in information of estimating parameter values $\bt$ using the Bayesian posterior distribution conditional on responses $\by$ collected via design $X$. The Bayesian expected utility is 
\begin{equation}
U_\mathcal{B}(X) = \mathrm{E}_{\mathcal{P}_T} \left\{ \mathrm{E}_{\mathcal{S}(\btheta_T, X)} \left[ u_{\mathcal{B}}(\btheta_T, \by,X) \right] \right\}.
\label{eqn:exputil}
\end{equation}
In (\ref{eqn:exputil}), the utility is evaluated at the true parameter values $\btheta_T$, indicating that the experimental aim is to estimate these values. The inner expectation is with respect to the responses $\by$ under $\mathcal{S}(\btheta_T; X)$, and the outer expectation with respect to the true parameter values $\btheta_T$, under the prior distribution, $\mathcal{P}_{\mathcal{T}}$.

A Bayesian optimal design is given by maximising $U_\mathcal{B}(X)$ over the design space $\mathbb{D} = \mathbb{X}^n$.

Two commonly-used utility functions are negative squared error (NSE; see, e.g., \citealt{robert_2007}, Section 2.5.1) and Shannon information (SH; see, e.g., \citealt{chalonerverdinelli1995}). The NSE utility is
$$u_{\mathcal{B},NSE}(\bt, \by,X) = - \left\Vert \bt - \mathrm{E}_{\mathcal{B}(\by; X)}\left(\btheta_T \right) \right\Vert_{I_p},$$
where $\mathrm{E}_{\mathcal{B}(\by; X)}\left(\btheta_T \right)$ is the Bayesian posterior mean of $\btheta_T$ and $\Vert \bu \Vert_{V} = \bu^\T  V \bu$ for a vector $\bu$ and a matrix $V$ of appropriate dimensions. The SH utility is
$$u_{\mathcal{B},SH}(\bt, \by,X) = \log \pi_{\mathcal{B}}(\bt \vert \by, X).$$
The Bayesian expected NSE and SH utilities are equal to the negative expected trace of variance matrix and entropy, respectively, of the Bayesian posterior distribution of $\btheta_T$, where expectation is with respect to the marginal distribution of $\by$, under the statistical model and prior distribution $\mathcal{P}_T$. Thus maximising the Bayesian expected NSE and SH utilities is equivalent to minimising expected uncertainty in the posterior distribution of $\btheta_T$, as measured by variance and entropy, respectively.



Notwithstanding the conceptual advantages of Bayesian optimal design listed in Section~\ref{sec:intro}, the disadvantage of finding Bayesian optimal designs in practice is clear. Firstly, the utility depends on $\by$ through the Bayesian posterior distribution. For most cases, this distribution is not available in closed form meaning the utility will not be available in closed form. Second, the integration required to evaluate the Bayesian expected utility given by (\ref{eqn:exputil}) will also not be analytically tractable. Although it is straightforward to derive a numerical approximation to the Bayesian expected utility $U_{\mathcal{B}}(X)$, this approximation then needs to be maximised over the $nk$-dimensional design space $\mathbb{D}$. The dimensionality of $\mathbb{D}$ can be high if the number of runs $n$ and/or number of controllable variables $k$ are large. This limits the effectiveness of commonly used numerical approximations such as, for example, Monte Carlo integration. However, in the last decade significant progress has been in the development of new computational methodology to find Bayesian optimal designs for a range of problems. See, for example, the review papers of \citet{ryanetal2016} and \citet{rainforth_etal_2023}, and references therein. 

A more fundamental disadvantage is that Bayesian optimal design assumes that the statistical model is correct, i.e. there exist $\btheta_T$ such that $\mathcal{T}(X) = \mathcal{S}(\btheta_T; X)$. This can be observed in the expected utility in (\ref{eqn:exputil}), where the utility is evaluated at $\btheta_T$, and expectation is taken with respect to the prior distribution of $\btheta_T$. When the statistical model is incorrect, there does not exist $\btheta_T$ such that $\mathcal{T}(X) = \mathcal{S}(\btheta_T; X)$. This makes evaluating the utility at $\btheta_T$ in (\ref{eqn:exputil}) nonsensical. However, it is well known \citep[see, e.g.][]{berk_1966, walker_2013, bissiri_etal_2016} that when the model is incorrect, the Bayesian posterior provides coherent inference about the parameter values, $\btheta_{SI}$, that minimise the Kullback-Liebler divergence between $\mathcal{T}(X)$ and $\mathcal{S}(\cdot; X)$. This raises the possibility of, in (\ref{eqn:exputil}), replacing evaluation of the utility at $\btheta_T$ by evaluation at $\btheta_{SI}$. The utility would then be representing gain in information on $\btheta_{SI}$, rather than the non-existent $\btheta_T$.
In this paper, we go further and consider a more general framework for optimal design revolving around Gibbs inference (of which Bayesian inference is a special case). 

\subsection{Gibbs inference} \label{sec:gibbsinf}

Gibbs inference (also known as general Bayesian inference) is an approach to statistical inference which is Bayesian-like, i.e. information is summarised by a probability distribution for unknown quantities, but does not necessarily rely on specification of a statistical model. Instead a loss function, denoted $\ell(\bt; \by, X)$ is specified. This function identifies desirable parameter values for observed responses $\by$. Notably, the corresponding classical M-estimators \citep[see, e.g., ][Chapter 7]{hayashi_2000} are defined
$$\hat{\btheta}_{\ell}(\by; X) = \arg \min_{\bt \in \Theta} \ell(\bt; \by, X).$$
The Gibbs posterior distribution, denoted $\mathcal{G}_\ell(\by; X)$, is given by
\begin{equation}
\pi_{\ell}\left(\bt \vert \by, X\right) \propto \exp \left[ - w \ell(\bt; \by, X) \right] \pi_{\ell}(\bt),
\label{eqn:gibbs}
\end{equation}
where $\pi_{\ell}(\cdot)$ is the pdf of the prior distribution (denoted $\mathcal{P}_\ell$) for the target parameter values (discussed below) and $w>0$ is a calibration weight (also discussed below). In (\ref{eqn:gibbs}), the quantity $\exp \left[ - w \ell(\bt; \by, X) \right]$ is known as the generalised likelihood. Similar to Bayesian inference, we assume that the prior distribution $\mathcal{P}_\ell$ does not depend on the design $X$.

\citet{bissiri_etal_2016} showed that the Gibbs posterior distribution provides coherent inference about the target parameter values defined as
$$\btheta_{\ell} = \arg \min_{\bt \in \Theta} L_{\mathcal{T}(X)}(\bt),$$
where $L_{\mathcal{T}(X)}(\bt) = \mathrm{E}_{\mathcal{T}(X)} \left[ w \ell(\bt; \by, X) \right]$ is the expected loss with respect to the responses $\by$ under the true response-generating distribution, $\mathcal{T}(X)$. 

Suppose the responses and design can be partitioned as $\by^\T  = \left(\by_1^\T ,\by_2^\T \right)^\T $ and $X = \left(X_1^\T ,X_2^\T \right)^\T $. Coherent inference, in this case, means that the Gibbs posterior formed from the complete data, i.e. by using generalised likelihood 
$\exp \left[ - w \ell(\bt; \by, X) \right]$ with prior $\mathcal{P}_{\ell}$ is identical to the Gibbs posterior from first evaluating the Gibbs posterior $\mathcal{G}_{\ell}(\by_1;X_1)$ and then using this as a prior distribution with generalised likelihood $\exp \left[ - w \ell(\bt; \by_2, X_2) \right]$.


The calibration weight, $w$, controls the rate of learning from prior, $\mathcal{P}_{\ell}$, to Gibbs posterior, $\mathcal{G}_{\ell}(\by; X)$. To see this, consider the extreme cases. As $w \to 0$, then $\mathcal{G}_\ell(\by;X) \to \mathcal{P}_{\ell}$, and, as $w \to \infty$, then $\mathcal{G}_\ell(\by;X)$ converges to a point mass at $\hat{\btheta}_\ell(\by; X)$. Therefore, the specification of the calibration weight is crucial for Gibbs inference.

Developing methodology for the automatic specification of the calibration weight is an active area of research \citep[see, e.g.][]{catoni_2007, grunwald_2012, bissiri_etal_2016, holmes_walker_2017, lyddon_etal_2019, jewson_2022}. For example, \cite{bissiri_etal_2016} described various approaches to the specification of the calibration weight. These include (a) choosing $w$ so that the prior distribution $\mathcal{P}_{\ell}$ contributes a fraction (e.g. $1/n$) of the information in the generalised likelihood (b) choosing $w$ to maintain operational characteristics, i.e. so that properties of the Gibbs posterior match those of the classical M-estimators; or (c) allowing $w$ to be unknown with a hierarchical prior distribution. This paper does not advocate for any singular approach for specifying the calibration weight. The Gibbs optimal design framework, as introduced in the next section, does not rely on any particular calibration method and can be used, in principle, with any method.

In summary, the prior distribution, $\mathcal{P}_{\ell}$, for the target parameter values, $\btheta_{\ell}$, is updated to the Gibbs posterior by using information in the generalised likelihood $\exp \left[ - w \ell(\bt; \by, X) \right]$. This is analogous to how the prior on the true parameter values is updated to the Bayesian posterior by using information in the likelihood. Note that the prior distribution for Gibbs inference can be non-informative reflecting a lack of prior knowledge about $\btheta_{\ell}$ in a similar way to how we can specify a non-informative prior on the true parameter values.

%

\subsection{Bayesian inference as Gibbs inference}

The Bayesian posterior distribution as a Gibbs posterior can be obtained under the self-information loss, i.e. the negative log-likelihood
$$\ell_{SI}\left(\bt ; \by, X\right) = - \log \pi\left( \by \vert \bt ; X\right),$$
with $w=1$, i.e. $\mathcal{G}_{SI}(\by;X) = \mathcal{B}(\by;X)$. In this case, the M-estimators, $\hat{\btheta}_{SI}(\by;X)$, are the maximum likelihood estimators. The corresponding expected loss is
$$L_{SI,\mathcal{T}(X)}(\bt) = \mathrm{E}_{\mathcal{T}(X)} \left[ - \log \pi\left( \by \vert \bt ; X\right) \right],$$
which, up to a constant independent of $\bt$, is the Kullback-Leibler (KL) divergence between the true response-generating distribution and the statistical model. Thus the target parameter values, denoted $\btheta_{SI}$, minimise this divergence. Under a correctly specified statistical model, where there exist $\btheta_T$ such that $\mathcal{S}(\btheta_T; X) = \mathcal{T}(X)$, then $\btheta_{SI} = \btheta_T$ with a minimised Kullback-Leibler divergence of zero. Now consider the case of a misspecified statistical model, where there do not exist parameter values $\btheta_T$ such that $\mathcal{S}(\btheta_T; X) = \mathcal{T}(X)$. In this case, Bayesian inference still provides coherent inference about $\btheta_{SI}$.

\section{Gibbs optimal design of experiments} \label{sec:god}

\subsection{Expected utility under the design distribution}

The proposal is to extend the Bayesian optimal design framework to Gibbs statistical inference. Similar to Bayesian optimal design, the Gibbs optimal design framework will involve maximising an expected utility function over the design space $\mathbb{D}$. The challenge is that, with Bayesian optimal design, the expectation of the utility is taken with respect to the responses and true parameter values, $\btheta_T$, under a joint distribution given by the statistical model and prior distribution for $\btheta_T$. Neither, a statistical model nor true parameter values necessarily exist under a Gibbs inference analysis.

We propose that the utility function be evaluated at the target parameter values. As defined in Section~\ref{sec:gibbsinf}, the target parameter values depend on the true response-generating distribution. However, this is unknown, so instead we replace these by the target parameter values under a distribution we term the \emph{designer distribution}. Furthermore, the expectation of the utility function is then taken under this designer distribution to give the objective function.

The purpose of the designer distribution is to represent the unknown true response-generating distribution $\mathcal{T}(X)$, in a similar way to the statistical model $S(\bt, X)$. This prompts the question: why not use the designer distribution as the statistical model? However, note that, post-experiment, the designer distribution will not be used as a statistical model for the observed responses. This means that it does not need to be ``useful", i.e. it does not need to be able to attain the aims of the experiment. Consider the simple example in Section~\ref{sec:intro}. A reasonable designer distribution is the unique-treatment model.  Every unique vector of controllable variables (i.e. treatment) is allowed a unique mean response with no relationship between the mean response at different controllable variables. This should provide a flexible representation of the true response-generating distribution but would be ineffective in learning the relationship between the controllable variables and response.  We derive expected utilities under this designer distribution in Section~\ref{sec:LM}.

Formally, let $\mathcal{D}(\br;X)$ denote the designer distribution, which depends on hyper-variables $\br$. The idea is that we may formulate more than one approximation to $\mathcal{T}(X)$, indexed by different hyper-variables. A hyper-variable distribution, denoted $\mathcal{C}$, is specified for $\br$. Similar to the prior distributions $\mathcal{P}_T$ and $\mathcal{P}_\ell$, the hyper-variable distribution is assumed not to depend on design $X$.

The target parameter values under the designer distribution are
$$\btheta_{\ell,\mathcal{D}}(\br;X) = \arg \min_{\bt \in \Theta} \mathrm{E}_{\mathcal{D}(\br;X)} \left[ w \ell(\bt; \by, X) \right]$$
and are a function of the hyper-variables $\br$ and design $X$.

The utility function is denoted $u_{\mathcal{G}}(\bt, \by,X)$. This can have the same functional form as a utility used under Bayesian optimal design. The difference is dependence on responses $\by$ and design $X$ is via the Gibbs posterior (instead of the Bayesian posterior).

The Gibbs expected utility objective function is then given by
\begin{equation}
U_{\mathcal{G}}(X) = \mathrm{E}_{\mathcal{C}} \left\{ \mathrm{E}_{\mathcal{D}(\br;X)} \left[ u(\btheta_{\ell,\mathcal{D}}(\br;X), \by,X) \right] \right\}.
\label{eqn:expGibbsutil}
\end{equation}
The inner expectation is with respect to the responses, $\by$, under the designer distribution, $\mathcal{D}(\br;X)$ and the outer expectation is with respect to the hyper-variables, $\br$, under $\mathcal{C}$. Note that the utility is evaluated at the target parameter values $\btheta_{\ell,D}(\br;X)$ rather than the true parameter values $\btheta_{T}$. This reflects the fact that the experimental aim is to estimate the target parameter values.

Bayesian optimal design is a special case of Gibbs optimal design, under the self-information loss and the designer distribution being the statistical model. This leads to the target parameters and hyper-variables being the true parameter values. Finally $\mathcal{C}$ is chosen as the prior distribution for the true parameter values. However, this tells us that the use of Bayesian optimal design implicitly assumes that the statistical model is true, i.e. a very strong assumption.

In the literature, we are aware of two modifications of Bayesian optimal design which are also special cases of Gibbs optimal design. Firstly, \cite{EtzioniKadane1993} considers the case where the outer expectation in the Bayesian expected utility (\ref{eqn:exputil}) is taken under a different prior distribution to that used to form the Bayesian posterior. This can represent the scenario where two separate individuals are (a) designing the experiment and (b) analysing the observed responses, and who may have differing prior beliefs on the true parameter values. In this case, the loss is self-information, and the designer distribution coincides with the statistical model. However, $\mathcal{P}_\ell$ is the prior distribution representing the prior beliefs of the individual analysing the observed responses, and $\mathcal{C}$ is the prior distribution representing the prior beliefs of the individual designing the experiment.

\cite{overstall_mcgree_2022} considered a Bayesian optimal design framework whereby the inner expectation in the Bayesian expected utility (\ref{eqn:exputil}), is taken under an alternative model for the responses. The motivation for doing so was to introduce robustness into the Bayesian optimal design process, for example, by making the alternative model more complex than the statistical model. The proposal in this paper is to extend this further to where the inference does not have to be Bayesian.

\subsection{Computational approach} \label{sec:comp}

In general, it will not be possible to find closed form expressions for the Gibbs expected utility, given by (\ref{eqn:expGibbsutil}). This is because the Gibbs posterior distribution and target parameter values will not be available in closed form. This matches how the Bayesian expected utility, given by (\ref{eqn:exputil}), is also typically not available in closed form.

In this section, we outline an exemplar computational approach that can be used to approximately find Gibbs optimal designs. It essentially uses a combination of a normal approximation to approximate the Gibbs posterior and a Monte Carlo integration to approximate the Gibbs expected utility. The approximate Gibbs expected utility is then maximised using the approximate coordinate exchange (ACE; \citealt{overstall_woods_2017}) algorithm. In the literature, the combination of Monte Carlo and normal approximations has been used to find Bayesian designs \citep{long_etal_2013}, including in partnership with ACE \citep{overstall_etal_2018}. It is anticipated that many existing computational approaches used to find Bayesian optimal designs can be repurposed to find Gibbs optimal designs (see Section~\ref{sec:disc} for more discussion). 

\subsubsection{Approximating the expected utility}

Let $\tilde{\btheta}_{\ell}(\by;X) = \arg \max_{\bt \in \Theta} \left[ - w\ell(\bt; \by, X) + \log \pi_{\ell}(\bt) \right]$ denote the Gibbs posterior mode and let 
\begin{eqnarray*}
\tilde{\Sigma}_{\ell}(\by;X) & = & - \left[ - w\frac{\partial \ell(\tilde{\btheta}_{\ell}(\by;X); \by, X)}{\partial \bt \partial \bt^\T } \right. \nonumber \\ 
& & \left. + \frac{\partial \log \pi_{\ell}(\tilde{\btheta}_{\ell}(\by;X))}{\partial \bt \partial \bt^\T } \right]^{-1},
\end{eqnarray*}
denote the negative inverse Hessian matrix of the log Gibbs posterior pdf evaluated at the Gibbs posterior mode. Then the Gibbs posterior distribution is approximated by a normal distribution with mean $\tilde{\btheta}_{\ell}(\by;X)$ and variance $\tilde{\Sigma}_{\ell}(\by;X)$. This is analogous to normal approximations used for Bayesian posterior distributions \citep[see, e.g.][page 237]{ohaganforster2004}. In the examples in Section~\ref{sec:examples}, we use a quasi-Newton method to numerically evaluate the Gibbs posterior mode. 

Now, due to the tractability of the normal distribution, approximations for many utility functions ensue. We denote such approximations by  $\tilde{u}(\bt, \by, X)$. For example, approximations for the NSE and SH utilities are
\begin{eqnarray*}
& & \tilde{u}_{NSE}(\bt, \by, X) = - \Vert \bt - \tilde{\btheta}_{\ell}(\by;X) \Vert_2^2,\\
& & \tilde{u}_{SH}(\bt, \by, X)  =  - \frac{p}{2} \log \left(2 \pi \right) - \frac{1}{2} \log \vert \tilde{\Sigma}_{\ell}(\by;X) \vert \\
& & - \frac{1}{2} \left( \bt - \tilde{\btheta}_{\ell}(\by;X)\right)^\T  \tilde{\Sigma}_{\ell}(\by;X)^{-1} \left( \bt - \tilde{\btheta}_{\ell}(\by;X)\right),
\end{eqnarray*}
respectively.

A Monte Carlo approximation to the expected utility is then given by
$$
\hat{U}_G(X) = \frac{1}{B} \sum_{b=1}^B \hat{u}(\btheta_{\ell, D}(\br_b,X), \by_b, X),
$$
where $\left\{ \br_b, \by_b \right\}_{b=1}^B$ is a sample from the joint distribution of responses $\by$ and hyper-variables $\br$ implied by designer distribution $\mathcal{D}(\br, X)$ and $\mathcal{C}$. 

For clarity, the following algorithm can be used to form the Monte Carlo approximation to the expected Gibbs utility.

\begin{enumerate}[leftmargin=*]
\item
Inputs are design, $X$; loss function, $\ell(\bt; \by,X)$; designer distribution, $\mathcal{D}(\br;X)$, hyper-variable distribution, $\mathcal{C}$; and Monte Carlo sample size $B$.
\item
For $b = 1,\dots,B$, complete the following steps.
\begin{enumerate}
\item
Generate hyper-variables $\br_b \sim \mathcal{C}$.
\item
Determine target parameter values under designer distribution $\mathcal{D}(\br_b;X)$, i.e.
$$\btheta_{\ell,\mathcal{D}}(\br_b;X) = \arg \min_{\bt \in \Theta} \mathrm{E}_{\mathcal{D}(\br_b;X)}\left[\ell(\bt; \by,X)\right].$$
\item
Generate responses $\by_b \sim \mathcal{D}(\br_b;X)$.
\item
Find Gibbs posterior mode
$$
\tilde{\btheta}(\by_b;X) = \arg \max_{\bt \in \Theta} \left[ - w\ell(\bt; \by_b, X) + \log \pi_{\ell}(\bt) \right]
$$
and 
\begin{eqnarray*}
& & \tilde{\Sigma}_{\ell}(\by_b;X) = - \left[ - w\frac{\partial \ell(\tilde{\btheta}_{\ell}(\by_b;X); \by_b, X)}{\partial \bt \partial \bt^\T } \right.\\
& & + \left. \frac{\partial \log \pi_{\ell}(\tilde{\btheta}_{\ell}(\by_b;X))}{\partial \bt \partial \bt^\T } \right]^{-1}.
\end{eqnarray*}
\item
Using the normal distribution 
$$\mathrm{N}\left[\tilde{\btheta}_{\ell}(\by_b;X), \tilde{\Sigma}_{\ell}(\by_b;X)\right]$$
as an approximation to the Gibbs posterior, form the approximation to the utility 
$$\tilde{u}_b = \tilde{u}_{\mathcal{G}}\left[\btheta_{\ell,\mathcal{D}}(\br_b;X), \by_b, X\right].$$
\end{enumerate}
\item
The Monte Carlo approximation to the Gibbs expected utility is
$$\tilde{U}_{\mathcal{G}} = \frac{1}{B}\sum_{b=1}^B \tilde{u}_b.$$
\end{enumerate}

In Step 2(b), we are required to find the target parameter values $\btheta_{\ell,\mathcal{D}}(\br_b;X)$ under the designer distribution $\mathcal{D}(\br_b;X)$. This is because the utility function in the inner expectation of the expected Gibbs utility (\ref{eqn:expGibbsutil}) is evaluated at the target parameter values. How this is accomplished will depend on the application but will typically require a numerical optimisation. We address this issue for each example in Section~\ref{sec:examples}.

\subsubsection{Maximising the approximate expected utility}

Unfortunately, the Monte Carlo approximation to the Gibbs expected utility, given by (\ref{eqn:exputil}) is not straightforward to maximise to yield the Gibbs optimal design. This is because it is computationally expensive to evaluate (one evaluation requires $B$ numerical maximisations to find the Gibbs posterior mode, and possible $B$ further numerical maximisations to find the target parameter values) and it is stochastic. To solve this problem, we use the ACE algorithm. Very briefly, ACE uses a cyclic descent algorithm (commonly called coordinate exchange in the design of experiments literature: see, e.g., \citealt{meyer_nachtsheim_1995}) to maximise a Gaussian process prediction of the expected utility, sequentially over each one-dimensional element of the design space. For more details on the ACE algorithm, see \cite{overstall_woods_2017}. 

While many methods could be used to maximise the approximate Gibbs expected utility, the ACE algorithm has been chosen due to its compatibility with any Monte Carlo approximation to the expected utility and an implementation being readily availability in software \citep{overstall_etal_2020}. 

\section{Linear models} \label{sec:LM}

\subsection{Introduction} \label{sec:lmintro}

In this section we consider Gibbs optimal design for the special case of linear models. The design of experiments for linear models is a significant topic in the literature. For textbook treatments, see, for example, \cite{atkinson_etal_2007}, \cite{goos_jones_2011} and \cite{morris_2011}. For specific treatments of robust design of experiments for linear models, see, for example, \cite{cook_nachtsheim_1982}, \cite{li_nachtsheim_2000}, \cite{bingham_chipman_2007}, \cite{tsai_gilmour_2010} and \cite{smucker_drew_2015}. One difference between the proposed framework and these papers, is that these papers attempt to find designs for a set of models. On the other hand, we acknowledge that the model we are fitting is incorrect and attempt to find a design to allow the fitted model to be as close to the truth as possible. In this way, our work is closer in spirit to, for example, \cite{waite_woods_2022} and \cite{azriel_2023}.

Suppose $y$ is a continuous response and the true data-generating process $\mathcal{T}(X)$ has
$$y_i = \mu(\bx_i) + e_i,$$
for $i=1,\dots,n$, where $\bx_i \in \mathbb{X} = [-1,1]^k$, and $e_1,\dots,e_n$ are independent and identically distributed random errors with mean zero and variance $\sigma^2 < \infty$. The true mean response, $\mu(\bx)$, for controllable variables $\bx$ is unknown and is approximated by a linear function $\bff(\bx)^\T  \bt$ where $\bff : \mathbb{X} \to \mathbb{R}^p$ is a specified regression function. 


Bayesian optimal design proceeds by assuming there exist $\btheta_T$ such that $\bff(\bx)^\T  \btheta_T = \mu(\bx)$ for all $\bx \in \mathbb{X}$. Under an improper uniform prior distribution, $\mathcal{P}_T$, for $\btheta_T$, the Bayesian expected NSE and SH utilities are proportional to the following two objective functions
\begin{equation} \label{eqn:classical}
\begin{split}
U^*_{\mathcal{B},NSE}(X) &=-\mathrm{tr}\left[ \left(F^\T  F\right)^{-1}\right]h_1(d) \\
U^*_{\mathcal{B},SH}(X) &= \log \left\vert F^\T  F \right\vert,
\end{split}
\end{equation}
respectively, where $F$ is an $n \times p$ matrix with $i$th column given by $\bff(\bx_i)^\T $, for $i=1,\dots,n$. The Bayesian optimal designs found by maximising the objective functions in (\ref{eqn:classical}) are also known as A- and D-optimal, respectively. These are commonly considered in the classical design of experiments literature.

We now consider Gibbs optimal design for the linear model for two different loss functions: sum of squares (Section~\ref{sec:SS}) and $L^2$ loss (Section~\ref{sec:L2}). In both cases, we specify the designer distribution to be the same as $\mathcal{T}(X)$ above. For the distribution of $e_1,\dots, e_n$, we specify a scale mixture of normal distributions, i.e. $e_i \vert \kappa \sim \mathrm{N}(0,\kappa)$, where $\kappa$ is a random variable with $\mathrm{E}_{\mathcal{C}}(\kappa) = \sigma^2$. When considering the marginal distribution of $e_i$, a scale mixture of normal distributions results in a wide range of symmetric, uni-modal distributions \citep[e.g.][]{west_1987} thus the designer distribution exhibits a large amount of flexibility. The hyper-variables are $\br = \left(\mu(\cdot), \sigma^2, \kappa\right)$.

\subsection{Sum of squares loss} \label{sec:SS}

Consider the sum of squares (SS) loss function
$$\ell_{SS}(\bt; \by, X) = \sum_{i=1}^n \left[ y_i - \bff(\bx_i)^\T  \bt\right]^2.$$
Under the SS loss, the M-estimators are the familiar least squares estimators $\hat{\btheta}_{SS} (\by; X)= \left(F^\T  F \right)^{-1} F^\T  \by$. 

We consider two different specifications for the calibration weight, $w$. In both cases, we specify a uniform prior distribution for the target parameter values, $\btheta_{SS}$.

\subsubsection{Fixed calibration weight based on maintaining characteristics} \label{sec:SSfixed}

Under a fixed calibration weight, the Gibbs posterior distribution for $\btheta_{SS}$ is 
$$\mathrm{N}\left[ \hat{\btheta}_{SS}(\by;X), \frac{1}{2w} \left(F^\T  F \right)^{-1}\right].$$
The calibration weight is chosen so that the Gibbs posterior variance is equal to the variance of $\hat{\btheta}_{SS}(\by; X)$. Therefore, $w=1/2\tilde{\sigma}^2$, where $\tilde{\sigma}^2$ is an estimate of the error variance under the unique-treatment model \citep[e.g.][]{gilmourtrinca2012}. Under this model, each unique treatment induces a unique mean response. Specifically, let $\bar{\bx}_1, \dots, \bar{\bx}_q$ be the unique treatments, i.e. the unique values of $\bx_1, \dots, \bx_n$, where $q \le n$. If $q<n$, i.e. there is at least one repeated treatment, then an estimator of $\sigma^2$ is given by
$$\tilde{\sigma}^2 = \frac{1}{n-q} \by^\T  \left(I_n - H_Z\right) \by,$$
where $H_Z = Z \left(Z^\T Z\right)^{-1} Z^\T $ and $Z$ is the $n \times q$ matrix with $ij$th element $Z_{ij} = 1$ if $\bx_i = \bar{\bx}_j$ and zero otherwise, for $i=1,\dots,n$ and $j=1,\dots,q$. This estimator is independent of the choice of regression function $\bff(\bx)$.

Under the designer distribution described in Section~\ref{sec:lmintro}, the target parameter values are 
\begin{equation}
\btheta_{SS, \mathcal{D}} (\br; X)= \left(F^\T  F \right)^{-1} F^\T  \bmu
\label{eqn:targetSS}
\end{equation}
where $\bmu = \left(\mu(\bx_1),\dots,\mu(\bx_n)\right)^\T $. See Section~\ref{sec:SM1} of the Appendix for a justification of (\ref{eqn:targetSS}). 

The Gibbs expected NSE and SH utilities are proportional to the following two objective functions
\begin{equation} \label{eqn:obj1}
\begin{split}
U^*_{\mathcal{G},NSE}(X) &=-\mathrm{tr}\left[ \left(F^\T F\right)^{-1}\right]h_1(d) \\
U^*_{\mathcal{G},SH}(X) &=- ph_2(d) + \log \left\vert F^\T F \right\vert,
\end{split}
\end{equation}
where $d=n-q$ is the pure error degrees of freedom \citep{gilmourtrinca2012}, $h_1(d) = 1$ if $d>0$ and $\infty$ if $d=0$, and $h_2(d) = \psi(d/2) - \log d + d/(d-2)$, if $d>0$ and $\infty$ if $d=0$, with $\psi(\cdot)$ the digamma function. Justification of the expressions in (\ref{eqn:obj1}) is given in Section~\ref{sec:SM2} of the Appendix.

The objective functions in (\ref{eqn:obj1}) can be seen as modifications of the objective functions for A- and D-optimality given in (\ref{eqn:classical}), respectively. In each case, the modification incentivises increasing $d$. In $U^*_{\mathcal{G},NSE}(X)$, the incentive given by $h_1(d)$ ensures the estimate $\tilde{\sigma}^2$ (and therefore the Gibbs posterior) exists. A corollary is that if the A-optimal design has at least one repeated design point, then the Gibbs optimal design under the NSE utility will be the A-optimal design. Although replication is a core principle of design of experiments \citep[e.g.][Section 1.2.4]{morris_2011}, whether or not an optimal design replicates design points is solely based on the objective function, through the balance of exploitation (replication) versus exploration (unique treatments). In $U^*_{\mathcal{G},SH}(X)$, the incentive for increasing $d$ is more significant. 

We consider a numerical example in Section~\ref{sec:SSexample}.

\subsubsection{Random calibration weight} \label{sec:SSrandom}

We now allow $w$ to be random. The target parameter values are unchanged from $\btheta_{SS, \mathcal{D}} (\br; X)$ in Section~\ref{sec:SSfixed}. We specify a prior distribution for the calibration weight given by $\pi(w) \propto w^{n/2 - 1}$. This choice of prior distribution results in the marginal Gibbs posterior distribution for $\btheta_{SS, \mathcal{D}} (\br; X)$ being the same as the corresponding Bayesian posterior distribution based on normally distributed errors. However, the formation of the Gibbs posterior does not require the assumption of normal errors. Specifically, the marginal Gibbs posterior distribution is a multivariate t-distribution \citep[e.g.][page 1]{kotznadarajah2004} with mean $\hat{\btheta}_{SS} (\by; X)$, scale matrix $\hat{\sigma}^2 \left(F^\T F\right)^{-1}$ and $n-p$ degrees of freedom, where $\hat{\sigma}^2 = \by^\T  \left(I_n - H_F\right)\by/(n-p)$ and $H_F = F\left(F^\T F\right)^{-1}F^\T $. Note that $\hat{\sigma}^2$ is the familiar regression function-dependent estimate of the error variance.

The Gibbs expected NSE and SH utilities are proportional to the following two objective functions
\begin{equation} \label{eqn:obj2}
\begin{split}
U^*_{\mathcal{G},NSE}(X) &=-\mathrm{tr}\left[ \left(F^\T F\right)^{-1}\right]h_1(d) \\
U^*_{\mathcal{G},SH}(X) &= \log \left\vert F^\T F \right\vert  - \mathrm{E}_{\mathcal{C}} \left( \mathrm{E}_{\mathcal{D}(\br, X)} \left\{ u_{SH}^*(\by; \br, X) \right\}\right),
\end{split}
\end{equation}
respectively, where 
$$u^*_{SH}(\by; \br, X) = p \log \hat{\sigma}^2 + n \log \left[ 1 + \frac{(n-p) \Vert \by - \bmu \Vert_{H_F}}{\hat{\sigma}^2} \right].$$
Note that the Gibbs optimal design under the NSE utility will be the A-optimal design. Similar to Section~\ref{sec:SSfixed}, the objective function under the SH utility is a modification of the objective function for D-optimality. The modification, $\mathrm{E}_{\mathcal{C}} \left( \mathrm{E}_{\mathcal{D}(\br, X)} \left\{ u_{SH}^*(\by; \br, X) \right\}\right)$, is not available in closed form.

Before we consider a numerical example, we reiterate our position from Section~\ref{sec:gibbsinf} that we do not advocate either of these two approaches for specifying the calibration weight. However, in Section~\ref{sec:SM3} of the Appendix we present the results of a simulation study which suggests that the fixed calibration weight of Section~\ref{sec:SSfixed} should be favoured based on the coverage and precision of probability intervals for the target parameter values.

\subsubsection{Numerical example} \label{sec:SSexample}

To illustrate the concepts, we consider a numerical example from \cite{gilmourtrinca2012}, where $n=16$, $k=3$ and 
$$\bff(\bx) = \left(1, x_1,x_2,x_3,x_1^2,x_2^2,x_3^2,x_1x_2,x_1x_3,x_2x_3\right)^\T $$
implementing an intercept, main and quadratic effects, and two-way interactions resulting in $p=10$.

For this setup, the A-optimal design has $d=2$, i.e. two repeated design points, meaning the Gibbs optimal design, under the NSE utility and fixed calibration weight is the A-optimal design. Therefore, we focus on Gibbs optimal designs under the SH utility. To find the design under the random calibration weight, for which the Gibbs expected utility is not available in closed form, we use the computational approach described in Section~\ref{sec:comp} but where the normal approximation is not required since the Gibbs posterior is available in closed form as a multivariate t-distribution. For this design, we also need to specify a hyper-variable distribution $\mathcal{C}$ for $\mu(\cdot)$, $\sigma^2$ and $\kappa$. We let $\mu(\cdot)$ be a realisation of a zero-mean Gaussian process with Mat{\'e}rn correlation function. Let $\bar{\bmu} = \left(\mu(\bar{\bx}_1),\dots,  \mu(\bar{\bx}_q)\right)^\T $ be the unique mean responses, with $\bmu = Z\bar{\bmu}$, then
\begin{equation}
\bar{\bmu} \sim \mathrm{N}\left(\bzero_q, \tau^2 A \right).
\label{eqn:GP}
\end{equation}
In (\ref{eqn:GP}), the $ij$th element of $A$ is $A_{ij}  =  c_1(\bx_i,\bx_j; \brho)$ with 
\begin{eqnarray*}
c_1(\bx_i,\bx_j; \brho)  &=&  \prod_{z=1}^k \left[ \left(1 + \rho \vert \bar{x}_{iz} - \bar{x}_{jz} \vert \right) \right. \\
& & \left. \qquad \times \exp \left( - \rho \vert \bar{x}_{iz} - \bar{x}_{jz} \vert \right) \right],
\end{eqnarray*}
being the Mat{\'e}rn correlation function. To complete, we let $\kappa$ be exponentially distributed with mean $\sigma^2 = 1$ and set $\rho_z = 1$ (for $z=1,\dots,k$), and $\tau^2 = 1$.

Table~\ref{tab:SS} shows the values of the objective functions (and efficiencies), under each design, for each design. The Gibbs optimal designs have identical performance under the two Gibbs expected utilities, indicating that the approach is invariable to the two different specifications of the calibration weight. These designs also have relatively good performance (85\% efficiency) compared to the D-optimal design. On the other hand, the D-optimal design has $d=0$ leading to a Gibbs expected utility, under fixed $w$, of $-\infty$ because the Gibbs posterior does not exist. The performance of the D-optimal design when compared against the Gibbs optimal design, under random $w$, is reasonable (69\% efficiency). Table~\ref{tab:SS} also shows the pure error degrees of freedom, $d$, for each design. Note that the Gibbs optimal designs have $d=6$, leading to $q=p=10$ unique design points. This is the minimum possible under a non-informative prior distribution for the target parameter values meaning that replication is as high as possible.

\begin{table}
\caption{For the squared error loss function, values of the objective function (and efficiencies), under each design, and for each design.} \label{tab:SS}
\begin{tabular}{llll} \hline
 & \multicolumn{3}{c}{Design} \\ \cline{2-4}
 & Fixed $w$ & Random $w$ & D-optimal \\
 & (Section~\ref{sec:SSfixed}) & (Section~\ref{sec:SSrandom}) & \\ \hline
Fixed $w$ & 11.95 & 11.95 & $-\infty$ \\
& (100\%) & (100\%) & (0\%) \\ \hline
Random $w$ & 15.73 & 15.73 & 12.05 \\
& (100\%) & (100\%) & (69\%) \\ \hline
D-optimal & 18.26 & 18.26 & 19.92 \\
& (85\%) & (85\%) & (100\%) \\ \hline
Pure error degrees & 6 & 6 & 0 \\
of freedom, $d$ & & & \\ \hline
\end{tabular}
\end{table}

\subsection{$L^2$ loss} \label{sec:L2}

\subsubsection{Derivation} \label{sec:L2deriv}

We now consider an alternative loss inspired by \cite{tuo_wu_2015b}. The loss is given by a quadratic approximation to $\int_{\mathbb{X}} \left[ \bar{\mu}(\bx; \by,X) - \bff(\bx)^\T  \btheta \right]^2 \mathrm{d}\bx$, where $\bar{\mu}(\bx; \by,X)$ is a non-parametric prediction of $\mu(\bx)$ based on $\by$ and $X$. Specifically, the loss is given by
\begin{equation}
\ell_{L^2}(\bt; \by, X) = \sum_{m=1}^M s_m \left[ \bar{\mu}(\bchi_m; \by,X) - \bff(\bchi_m)^\T  \bt \right]^2
\label{eqn:L2loss}
\end{equation}
where $\left\{ s_m, \chi_n\right\}_{m=1}^M$ are the quadrature weights and nodes, respectively, with $M$ being the number of quadrature nodes. If $\bar{\mu}(\bx; \by,X)$ is an unbiased predictor and the quadrature is exact, the target parameter values minimise the $L^2$ norm of the difference between $\mu(\bx)$ and $\bff(\bx)^\T \bt$. For this reason, the loss given by (\ref{eqn:L2loss}) is referred to as the $L^2$ loss.

For $\bar{\mu}(\bx;\by,X)$, we use the predictive mean from a non-zero nugget, zero-mean Gaussian process. We assume
$$\by \sim \mathrm{N}\left[\bzero_n, \phi^2 \left( \nu I_n + B \right) \right],$$
where $\nu>0$ is the nugget and the $ij$th element of $B$ is
\begin{eqnarray*}
B_{ij} &=& c_2(\bx_i,\bx_j; \balpha)\\
& = & \exp \left[ - \sum_{z=1}^k \alpha_z \left(x_{iz}-x_{jz}\right)^2 \right],
\end{eqnarray*}
for $i,j = 1,\dots,n$. The parameters $\phi^2$, $\nu$ and $\balpha = \left(\alpha_1,\dots,\alpha_k\right)^\T $ are estimated by leave-one-out cross validation, and these estimates are denoted $\bar{\phi}^2$, $\bar{\nu}$ and $\bar{\balpha}$, respectively. The predictive mean is given by
$$\bar{\mu}(\bx;\by,X) = \bb(\bx; \bar{\alpha})^\T  \bar{V} \by,$$
where $\bb(\bx; \bar{\alpha}) = \left[ c_2(\bx_1,\bx; \bar{\alpha}),\dots,c_2(\bx_n,\bx; \bar{\alpha})\right]^\T $, the $ij$th element of $\bar{B}$ is $c_2(\bx_i,\bx_j; \bar{\balpha})$, and $\bar{V} = \left(\bar{\nu}I_n + \bar{B}\right)^{-1}$.

The M-estimators are given by minimising (\ref{eqn:L2loss}) with respect to $\bt$ and are
$$\hat{\btheta}_{L^2}(\by;X) = E_Q^{-1} F_Q^\T  \bar{D}_Q \bar{V}\by,$$
where $E_Q = \sum_{m=1}^M s_m \bff(\bchi_m)\bff(\bchi_m)^\T $, $F_Q$ is an $M \times p$ matrix with $m$th row $s_m \bff(\bchi_m)^\T $ and $\bar{D}_Q$ is an $M \times n$ matrix with $mi$th element $c_2(\bchi_m,\bx_i; \bar{\balpha})$, for $m=1,\dots,M$ and $i=1,\dots,n$.

We consider fixing $w$ to maintain operational characteristics. Under a fixed $w$, the Gibbs posterior distribution of the target parameter values is 
\begin{equation}
\mathrm{N}\left[\hat{\btheta}_{L^2}(\by;X), \frac{1}{2w}E_Q^{-1} \right].
\label{eqn:L2post}
\end{equation}
We specify $w$ such that the trace of the Gibbs posterior variance is approximately equal to the trace of the variance of $\hat{\btheta}_{L^2}(\by;X)$. Specifically, 
$$w = \frac{\mathrm{tr} \left[ E_Q^{-1} \right]}{2 \bar{\sigma}^2 \mathrm{tr}\left[ E_Q^{-1} F_Q^\T  \bar{D}_Q \bar{V}\bar{V}\bar{D}_Q^\T  F_Q E_Q^{-1} \right]},$$
where
$$\bar{\sigma}^2 = \frac{ \left\Vert \left(I_n - \bar{B}\bar{V} \right)\by \right\Vert_{I_n}}{\mathrm{tr}\left\{ \left( I_n - \bar{V} \bar{B}^\T\right)\left( I_n - \bar{B}\bar{V} \right)\right\}},$$
is an estimate of the error variance under the Gaussian process prediction.

The target parameter values, $\btheta_{L^2, \mathcal{D}}(\br; X)$, are given by minimising
$$\mathrm{E}_{\mathcal{D}}(\br; X)\left[ w \ell_{L^2}(\bt; \by, X) \right],$$
with respect to $\bt$. These values are not available in closed form, so we approximate, via the delta method, to obtain
$$\btheta_{L^2, \mathcal{D}}(\br; X) = E_Q^{-1} F_Q^\T  \bar{D}_Q \bar{V}\bmu.$$

The Gibbs expected NSE and SH utilities give the following two objective functions
\begin{equation} \label{eqn:obj3}
\begin{split}
U^*_{\mathcal{G},NSE}(X) &=  \mathrm{E}_{\mathcal{C}} \left\{ \mathrm{E}_{\mathcal{D}(\br;X)} \left[ u^*_{NSE}(\by; \br, X) \right] \right\} \\
U^*_{\mathcal{G},SH}(X) &= \log \pi_{\mathcal{G}} \left[ \btheta_{L^2, \mathcal{D}}(\br; X) \vert \by, X\right]
\end{split}
\end{equation}
where $u^*_{NSE}(\by; \br, X) = \Vert \btheta_{L^2, \mathcal{D}}(\br; X) -  \hat{\btheta}_{L^2}(\by;X) \Vert$ and $\pi_{\mathcal{G}} \left[ \bt \vert \by, X\right]$ is the pdf of (\ref{eqn:L2post}). Neither of the two expressions in (\ref{eqn:obj3}) are available in closed form.

\subsubsection{Numerical example} \label{sec:L2example}

We return to the numerical example in Section~\ref{sec:SSexample} using the same designer distribution. Note that the correlation functions under the designer distribution and the Gaussian process prediction, $\bar{\mu}(\bx;\by,X)$, are different. 

For the quadrature, we use a Gauss-Legendre scheme \citep{mvquad} with $M=1000$. Since the two objective functions given by (\ref{eqn:obj3}) are not available in closed form we use the computational methodology of Section~\ref{sec:comp} but where the normal approximation is not required since the Gibbs posterior is available in (\ref{eqn:L2post}). 

Figure~\ref{fig:LM_fig_L2} shows two-dimensional projections of the Gibbs optimal designs under the NSE and SH utilities. They appear similar to space-filling designs, which motivates a comparison with a maximum projection Latin hypercube design \citep{joseph_etal_2015}. Therefore, also shown in Figure~\ref{fig:LM_fig_L2}, are two-dimensional projections of a maximum projection Latin hypercube design (LHD) with $n=16$ and $k=3$ (where the design points are scaled to $\mathbb{X}$). Inspection of Figure~\ref{fig:LM_fig_L2} shows that the Gibbs optimal designs typically have points closer to the extremes of $\mathbb{X}$ than the LHD. Table~\ref{tab:L2} shows the values (and efficiencies) of the objective function under each design. The LHDs peform reasionably well compared to the Gibbs optimal designs (efficiencies of 70-75\%). The Gibbs optimal design under the SH utility performs well (88\%) when compared against the design under the NSE utility. However, the converse is not true: the Gibbs optimal design under the NSE utility has only 45\% efficiency under the expected SH utility function. Figure~\ref{fig:LM_fig_L2} shows that the NSE design has clusters of designs points and it appears that these result in the observed lack of SH-efficiency.

\begin{figure}
    \includegraphics[scale=0.9]{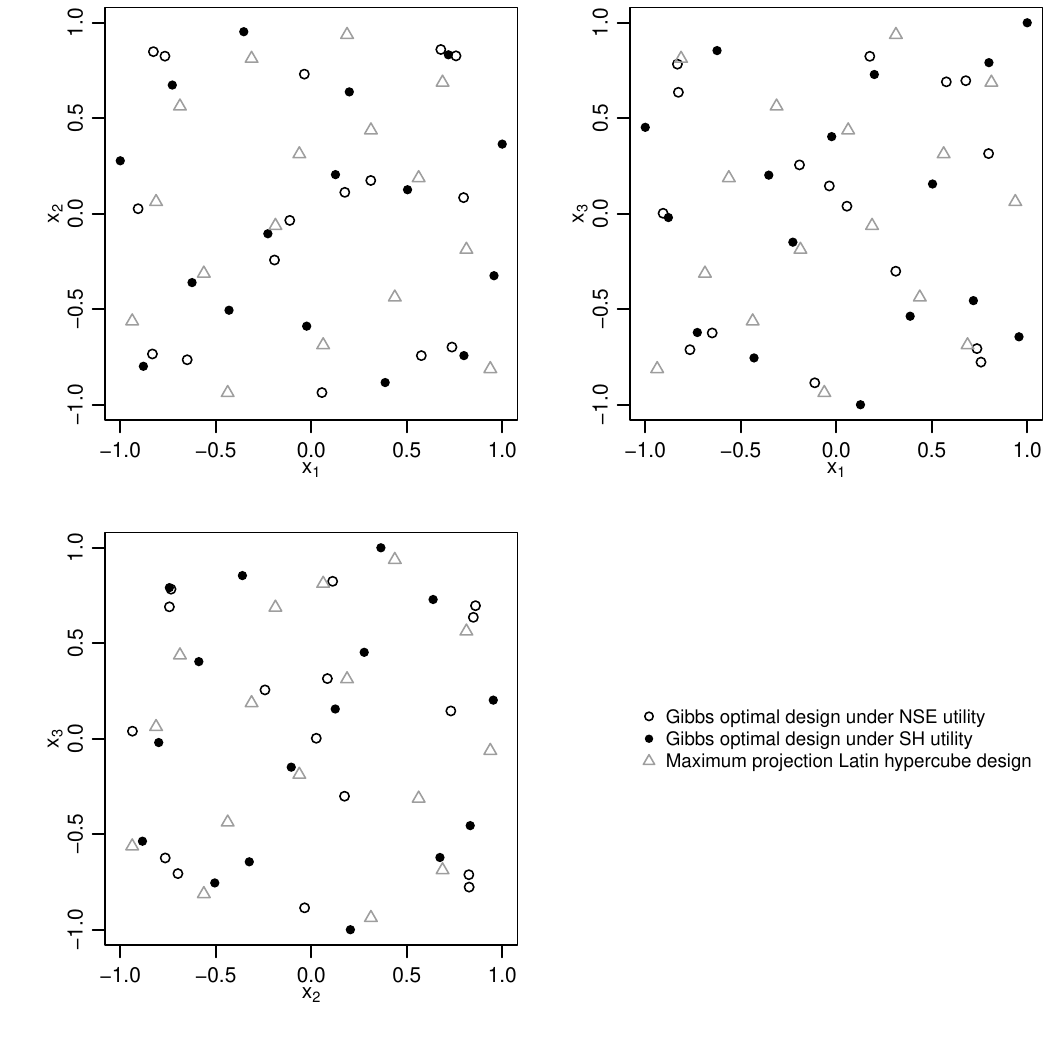}
    \caption{Plots of the Gibbs optimal designs, under the $L^2$ loss, and maximum projection Latin hypercube design.}
    \label{fig:LM_fig_L2}
\end{figure}

\begin{table}
\caption{For the $L^2$ loss function, values of the objective function (and efficiencies), under each Gibbs expected utility, and for each design, including the maximum projection Latin hypercube design (LHD).} \label{tab:L2}
\begin{tabular}{llll} \hline
 & \multicolumn{3}{c}{Design} \\ \cline{2-4}
 & NSE & SH & LHD \\
 & utility & utility & \\ \hline
NSE utility & -0.6950 & -0.7882 & -0.9191 \\
& (100\%) & (88\%) & (75\%) \\ \hline
SH utility & -12.21 & -4.299 & -7.932 \\
& (45\%) & (100\%) & (70\%) \\ \hline
\end{tabular}
\end{table}

\section{Examples} \label{sec:examples}

\subsection{Count responses and overdispersion} \label{sec:count}

In this example, the aim of the experiment is to learn the relationship between controllable variables and count responses. Design of experiments for count responses has previously been considered by, for example, \citet[Chapter 8]{myers_etal_2010}, \citet[Chapter 5]{russell_2019} and \cite{mcgree_eccleston_2012}. The typical approach is to assume the responses are realisations of a Poisson distribution. Accordingly, a Bayesian optimal design can then be found. Specifically, we may assume
\begin{eqnarray}
y_i & \sim & \mathrm{Poisson}\left[ \exp \left(\eta_i \right)\right] \label{eqn:poisson}\\
\eta_i & = & \bff(\bx_i)^\T  \bt, \label{eqn:poissoneta}
\end{eqnarray}
with $\bt = \btheta_T$, along with a prior distribution, $\mathcal{P}_T$, for $\btheta_T$. In this example, we suppose that $\bff(\bx) = \left(1,x_1,x_2\right)^\T $, i.e. an intercept and main effects for $k=2$ controllable variables, and $\mathbb{X} = [-1,1]^2$ 

However, assuming a Poisson distribution is a strong assumption. For example, one property of the Poisson distribution is that its mean and variance are equal, $\mathrm{var}(y_i) = \exp \left(\eta_i \right)$. In practice, this equality of mean and variance can often fail, a phenomenon known as over-dispersion \citep[see, e.g., ][Section 4.5]{pawitan_2013}.  

In the frequentist literature, a commonly-used approach is estimation via quasi-likelihood \citep[see, e.g., ][pages 512-517]{davison2003}. This does not require the assumption of a specified probability distribution for the response and allows the variance to be equal to the mean multiplied by a dispersion parameter $\phi$. 

The quasi log-likelihood is proportional to the log-likelihood under the statistical model given by (\ref{eqn:poisson}), but multiplied by a factor of $1/\phi$, i.e.
$$l_{QL}(\bt;\by,X) = \frac{1}{\phi} \sum_{i=1}^n \left[y_i \eta_i - \exp \left(\eta_i \right)\right],$$
where $\eta_i$ is given by (\ref{eqn:poissoneta}). This means that the maximum quasi-likelihood estimators are identical to the maximum likelihood estimators but their asymptotic variance is inflated by a factor of $\phi$, to account for over-dispersion. The dispersion parameter is estimated from the responses \citep[see, e.g.][page 483]{davison2003}. 

We can use Gibbs inference using the negative quasi log-likelihood as the loss function (but removing the factor of $1/\phi$), i.e.  
$$
\ell_{QL}(\bt; \by, X) = \sum_{i=1}^n \left[\exp \left(\eta_i \right) - y_i \eta_i\right],
$$
where $\eta_i$ is given by (\ref{eqn:poissoneta}). The calibration weight, $w$, actually corresponds to $1/\phi$. We specify $w$ to be the reciprocal of the estimator of $\phi$, i.e.
$$w = \frac{n-p}{\sum_{i=1}^n \left[ y_i - \exp(\hat{\eta}_i) \right]^2/\exp(\hat{\eta}_i)},$$
where $\hat{\eta}_i = \bff(\bx_i)^\T  \hat{\btheta}_{QL}$ with
$$\hat{\btheta}_{QL} = \arg \min_{\bt \in \Theta}\ell_{QL}(\bt; \by, X) = \arg \max_{\bt \in \Theta}l_{QL}(\bt; \by, X),$$
being the maximum (quasi) likelihood estimators.

For the designer distribution, we modify the unique-treatment model from Section~\ref{sec:LM} to the count response scenario. We suppose 
$$y_i \sim \mathrm{neg-bin}\left(\mu_i, \alpha_i \right),$$
for $i=1,\dots,n$, where $\mathrm{neg-bin}(\mu,\alpha)$ denotes a negative binomial distribution with mean $\mu$ and variance $\mu + \mu^2/\alpha$. The mean for the $i$th response is given by
\begin{equation}
\log \mu_i = \bff(\bx_i)^\T  \bbeta + \bz_i^\T  \btau,
\label{eqn:poissonmean}
\end{equation}
where $\bz_i = \left(Z_{i1}, \dots, Z_{iq}\right)^\T $ is the $q \times 1$ vector identifying which of the $q$ unique treatments is subjected to the $i$th run, and $\btau = \left(\tau_1,\dots,\tau_q\right)^\T $ is the $q \times 1$ vector of unique-treatment effects. The mean is chosen as (\ref{eqn:poissonmean}) so that the elements of $\btau$ represent the discrepancy between the true mean and the assumed mean under the quasi-likelihood specification (on the log scale). This allows a type of parity when we compare Bayesian and Gibbs optimal designs later in this sub-section.  

The scale parameter is $\alpha_i = \mu_i^2/(\kappa \mu_i - \mu_i)$, which results in $\mathrm{var}(y_i) = \kappa \mu_i$. Thus $\kappa>1$ controls the extent of over-dispersion.

Under the designer model, the hyper-variables are $\br = \left(\bbeta, \btau, \kappa\right)$. The target parameter values are
$$\btheta_{\ell_{QL},D}(X,\bbeta, \kappa,\btau) = \arg \min_{\bt \in \Theta} \sum_{i=1}^n \left[\exp \left(\eta_i \right) - \mu_i \eta_i\right],$$
where $\eta_i$ is given by (\ref{eqn:poissoneta}) and $\mu_i$ is given by (\ref{eqn:poissonmean}). Firstly, $\btheta_{\ell_{QL},D}(X,\bbeta,\theta,\btau)$ are independent of $\kappa$. Second, the target parameter values, $\btheta_{\ell_{QL},D}(X,\bbeta,\theta,\btau)$, under the designer distribution are the maximum likelihood likelihood estimators under a Poisson GLM. Therefore, in Step 2.(a) of the algorithm to approximate the Gibbs expected utility, in Section~\ref{sec:comp}, we can readily use Fisher scoring to calculate the target parameter values.

We find Bayesian and Gibbs optimal designs, for $n=10,20,\dots,50$, under the NSE utility function. For the Bayesian optimal design, we assume the prior distribution, $\mathcal{P}_T$, is such that the elements of $\btheta_T$ are independent, each having a $\mathrm{N}(0,1)$ distribution. For the Gibbs optimal design, for the hyper-variable distribution, we assume $\bbeta$, $\btau$ and $\kappa$ are independent. We assume the elements of $\btau$ are independent, each having a $\mathrm{U}\left(-2,2\right)$ distribution, and suppose that the elements of $\bbeta$ are independent, each having a $\mathrm{N}(0,1)$ distribution.  We specify $\kappa \sim U(1,5)$. 

\begin{figure}
    \includegraphics[scale=0.9]{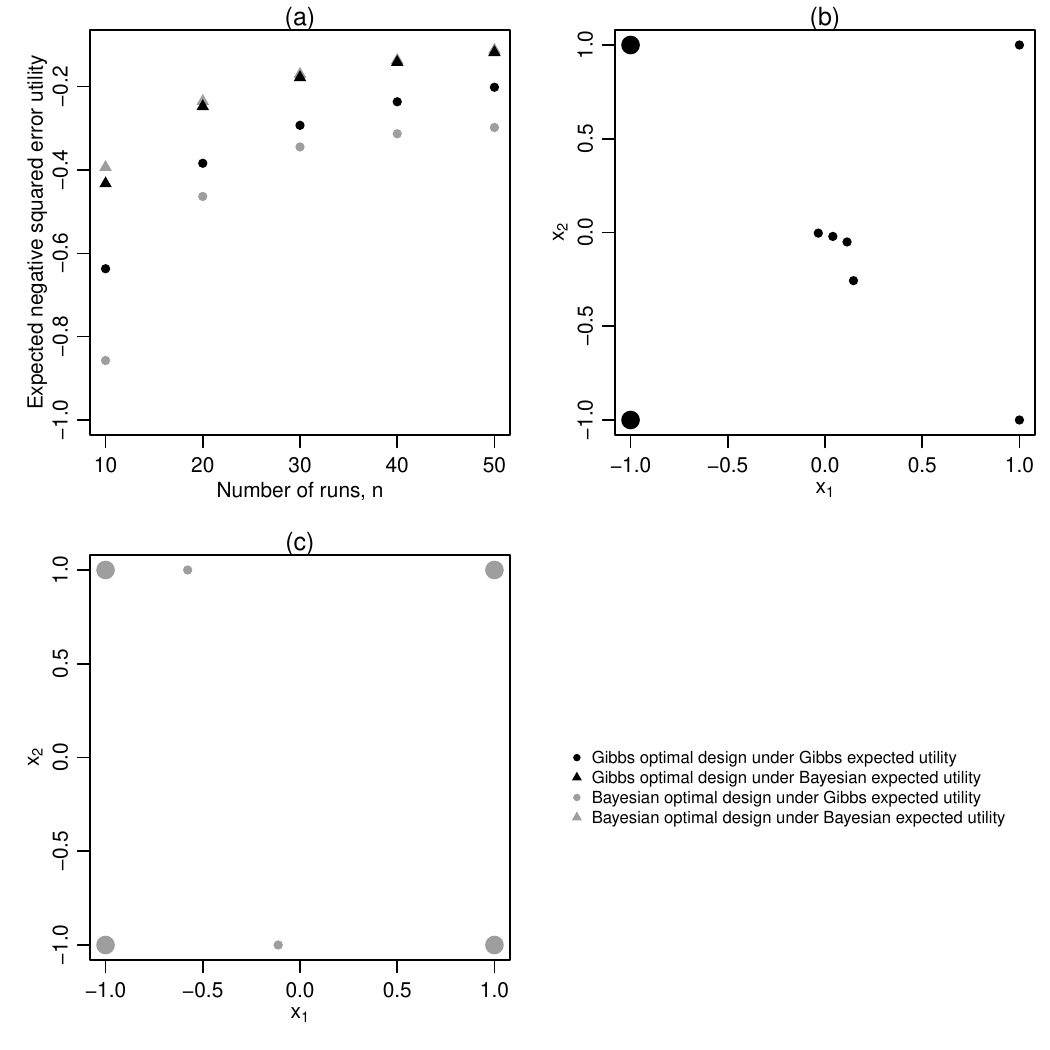}
    \caption{Plots for the count response experiment. Sub-panel (a) shows plots of the Gibbs and Bayesian expected negative squared error utility against number of runs, $n$, for the Gibbs and Bayesian optimal designs. Sub-panels (b) and (c) shows plots of $x_2$ against $x_1$ for the Gibbs and Bayesian optimal designs, respectively.}
    \label{fig:count}
\end{figure}

Figure~\ref{fig:count}(a) shows plots of the Gibbs expected negative squared error utility for both Gibbs (black circles) and Bayesian optimal designs (grey circles). It also shows the Bayesian expected negative squared error utility for both designs (triangles with Gibbs black and Bayesian grey). It can be seen that both designs are sub-optimal under the alternative expected utility. However, the difference on the scale of the expected Bayesian utility is larger than on the scale of the expected Gibbs utility. Informally, less is lost by not making strong assumptions and these assumptions being true, than by making these strong assumptions and these not being true.

Figures~\ref{fig:count}(b) and (c) shows plots of $x_2$ against $x_1$ for the Gibbs and Bayesian optimal designs, respectively. The size of the plotting character is proportional to the number of repeated designs points at each support point. Clearly, the Bayesian optimal design places points at the extremes of $\mathbb{X}$. Whereas, the Gibbs optimal design includes points in the interior. 

\subsection{Time-to-event responses and proportional hazards}

In this example, the aim of the experiment it to learn the relationship between controllable variables and a time-to-event response (possible subjected to non-informative right censoring). Previous approaches for designing these type of experiments have been proposed by, for example, \cite{chaloner_larntz_1992}, \cite{konstantinou_etal_2015}, and \cite{konstantinou_etal_2017}.

It is common for time-to-event responses to assume a proportional hazards model \citep[see, e.g., ][Section 10.8.2]{davison2003}. This is where the pdf for the response is
\begin{equation}
f(y_i; \bt, \bx_i) = h_0(y_i) \exp \left[ \eta_i - e^{\eta_i} \int_0^{y_i} h_0(u)\mathrm{d}u \right],
\label{eqn:ph}
\end{equation}
where $h_0(\cdot)$ is a baseline hazard and $\eta_i = \bff(\bx_i)^\T  \bt$. In this example, we suppose that $\bff(\bx) = \left(1,x_1,x_2,x_3\right)^\T $, i.e. an intercept and main effects for $k=3$ controllable variables. Bayesian inference requires that a particular form for the baseline hazard is assumed \citep[see, e.g., ][Chapter 2]{ibrahim_etal_2001}. By contrast, under the frequentist approach, the Cox proportional hazards model does not requires a particular form for the baseline hazard to be assumed, with estimators given by maximising the partial likelihood. 

\subsubsection{Gibbs optimal designs} \label{sec:PHgod}

We consider Gibbs inference using the negative partial log-likelihood as a loss function. \cite{bissiri_etal_2016} considered such an analysis for investigating the relationship of bio-markers and survival time for colon cancer. They set $w=1$ which can be informally justified as ensuring the asymptotic Gibbs posterior variance is equal to the asymptotic variance of the maximum partial likelihood estimators.

If we assume that all responses are unique (i.e. there are no tied responses), then the loss function is
$$\ell_{PL}(\by, \bc; \bt) = \sum_{i=1}^n c_i \left[ \log \sum_{j \in R_i} e^{\eta_j} - \eta_i \right],$$
where $\bc = (c_1,\dots,c_n)$, with $c_i = 0$ if the $i$th response is right censored at $y_i$, and $c_i = 1$ otherwise, and $R_i = \left\{j \in (1,\dots,n) \vert y_j \ge y_i \right\}$ is the risk set, i.e. the subset of $\left\{1,\dots,n\right\}$ of responses which have not been observed or right censored before $y_i$. 

For the designer distribution, we assume $y_i$ and $c_i$ are independent. For the censoring indicators, it is assumed that $c_i \sim \mathrm{Bernoulli}(\rho)$ for some $\rho \in [0,1]$ controlling the probability of right censoring, with $\rho=1$ corresponding to no censoring. If the designer distribution for $y_i$ is a proportional hazards distribution, i.e. has pdf given by (\ref{eqn:ph}), with $\eta_i = \bff(\bx_i)^\T  \bbeta$, then the target parameter values are $\btheta_{PL,D}(\br,\bc,X) = \bbeta$. Moreover, the maximum partial likelihood estimators and Gibbs posterior are invariant to the actual choice of baseline hazard. Therefore, for simplicity, we use the exponential distribution with $h_0(u) =1$. The hyper-variables are $\br = \left(\bbeta, \rho \right)$. The hyper-variable distribution is such that $\bbeta$ and $\rho$ are independent. The elements of $\bbeta$ are assumed independent with each element having a $\mathrm{N}(0,5)$ distribution. We fix $\rho$ and investigate sensitivity to its specification.

\begin{figure}
    \includegraphics[scale=0.9]{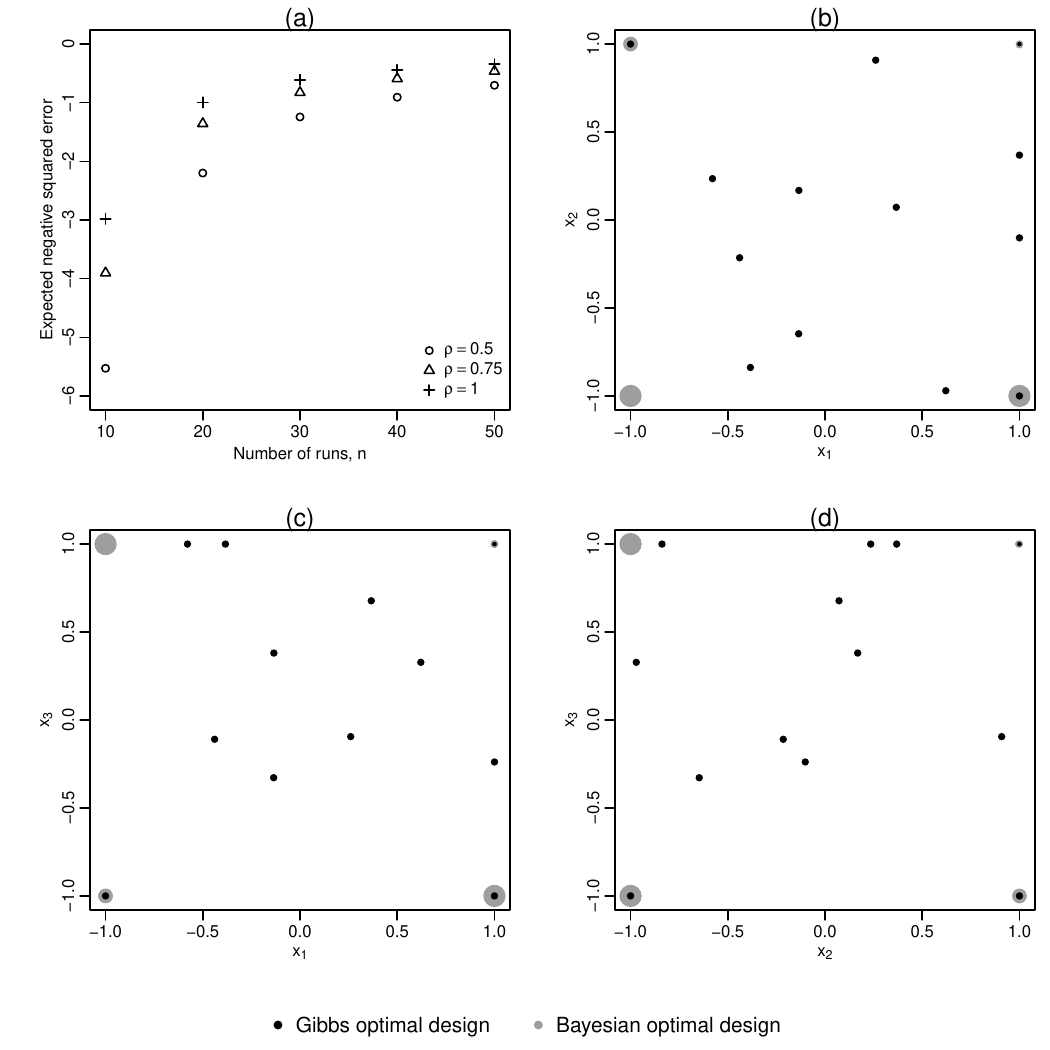}
    \caption{Panel (a): Plot of expected negative squared error utility against number of runs, $n$, for different levels of censoring (controlled by $\rho$). Panels (b)-(d): Plots of $x_1$, $x_2$ and $x_3$ for the Gibbs and Bayesian optimal designs. The size of the plotting character is proportional to the number of repetitions of the design points.}
    \label{fig:Cox}
\end{figure}

We consider number of runs, $n=10,20,\dots,50$, censoring probability $\rho = 0.5, 0.75, 1.00$ and the NSE utility function. Figure~\ref{fig:Cox}(a) shows a plot of Gibbs expected NSE utility against number of runs, $n$, for the different values of $\rho$. It makes intuitive sense, that as $\rho$ increases (probability of censoring decreases), the Gibbs expected utility increases since less information is lost. The effect of amount of the censoring on the actual design was investigated and it was found that the optimal design was only negligibly affected by probability of censoring. For example, the optimal design for $n=10$ and $\rho=0.5$ has a very similar value of expected Gibbs utility (with $\rho = 1$) as the optimal design found under $\rho = 1$.

We compare the Gibbs optimal designs to Bayesian optimal designs under a Weibull proportional hazards model. In this case, $h_0(u) = \psi u^{\psi-1}$, where $\psi >0$ is an unknown nuisance parameter. If $\psi > 1$ ($\psi<1$), then the baseline hazard is an increasing (decreasing) function of $u$. The log-likelihood is given by
\begin{eqnarray*}
& & \log \pi(\by , \bc \vert \bbeta, \psi; X)\\
& &  = \sum_{i=1}^n I(c_i=1) \left[ \eta_i + \log \psi + (\psi-1)\log y_i \right] - e^{\eta_i}y_i^{\psi}.
\end{eqnarray*}
The elements of $\bbeta$  and $\bc$ are assumed to have the same distributions as under the Gibbs expected utility above. The parameter $\psi$ is assumed to have a $\mathrm{U}[0.5,1.5]$ distribution. 

\begin{table}
    \centering
    \begin{tabular}{lrrrrr} \hline
     Number of runs, $n$ & 10 & 20 & 30 & 40 & 50 \\ \hline
     $\hat{U}_G(\bar{X}_B)/\hat{U}_G(\bar{X}_G)$ & 0.037 & 0.061 & 0.131 & 0.159 & 0.175 \\
     $\hat{U}_B(\bar{X}_G)/\hat{U}_B(\bar{X}_B)$ & 0.270 & 0.510 & 0.600 & 0.650 & 0.680 \\ \hline
    \end{tabular}
    \caption{Estimated expected utility efficiency of the Gibbs and Bayesian optimal designs under the alternative utility.}
    \label{tab:cens}
\end{table}

We present results for the case where the censoring probability is $\rho = 0.75$. The results for $\rho = 0.5$ and $\rho = 1$ were qualitatively the same. Table~\ref{tab:cens} shows the estimated efficiency of the Gibbs and Bayesian optimal designs under the alternative utility. Suppose for a given value of $n$ and $\rho$, $\bar{X}_G$ and $\bar{X}_B$ denote the Gibbs and Bayesian optimal design, respectively. Then, for example, the efficiency of $\bar{X}_G$ under the expected Bayesian utility is $U_B(\bar{X}_G)/U_B(\bar{X}_B)$. The conclusion is similar to the count response experiment in Section~\ref{sec:count}, in that less information is lost by not making strong assumptions (i.e. assuming a Weibull distribution) and these assumptions being true than by making these strong assumptions and these not being true.

Figure~\ref{tab:cens}(b)-(c) show projections of the design points. The size of the plotting symbol is proportional to the number of repeated design points. A similar phenomenon to the count response experiment in Section~\ref{sec:count} occurs, in that the Gibbs optimal design has more points in the interior.

\section{Discussion} \label{sec:disc}

This paper proposes a new design of experiments framework called Gibbs optimal design. This framework extends the decision-theoretic Bayesian optimal design of experiments approach to account for potentially misspecified statistical models. A computational approach is outlined for finding designs in practice. The framework is demonstrated on illustrative examples involving linear models, in addition to experiments involving count and time-to-event responses. From the examples, it was found that less information is lost by not making strong assumptions, and these being true, than by making strong assumptions and these assumptions not being true.

The key is the specification of the designer distribution. It needs to be flexible enough that it is plausible that the true response-generating distribution is a special case. In the illustrative examples, a unique-treatment model was employed as the designer distribution.

The framework should open up avenues for further research. One interesting area is that the target parameter values of Gibbs inference under a fixed design are dependent on the design. This can be viewed as an unattractive property, and future work will attempt to address this. One potential idea is that the utility function be evaluated at \emph{desired target parameter values}, where these are the limit of the target parameter values as $n \to \infty$ (and under certain conditions) and have physical interpretation independent of the design. As an example, consider the linear model in Section~\ref{sec:LM}. Under certain conditions, \cite{wong_etal_2017} showed that the target parameter values under the sum of squares loss converged to the parameter values that minimise the $L^2$ norm of the difference between the designer distribution mean response and $\bff(\bx)^\T  \bt$. These parameter values have desirable physical interpretation.

A further area of research is to investigate the sensitivity of the Gibbs optimal design to the method chosen to specify calibration weight. In the numerical example in Section~\ref{sec:SS}, the Gibbs optimal designs were insensitive to the different specifications of the calibration weight. This will need further investigation in consort with current research on automatic specifications.

Lastly, in Section~\ref{sec:comp}, a computational approach was outlined for finding Gibbs optimal designs by maximising the Gibbs expected utility. This was based on the simple premise of approximating the Gibbs posterior distribution by a normal distribution. This builds on existing approaches to find Bayesian optimal designs. We encourage researchers who develop computational methodology to find Bayesian designs to generalise to finding Gibbs optimal designs. There is an additional step in the latter case in that the target parameter values, $\btheta_{\ell,\mathcal{D}}(\br; X)$, need to be determined. This is exemplified by Step 2.(b) in the algorithm in Section~\ref{sec:comp}. However, this step adds negligible computational complexity relative the ultimate task of approximating and maximising the expected utility over a design space of, potentially, high dimensionality.

\appendix

\section{Justification of target parameter values given by equation (6)} \label{sec:SM1}

The following results will prove useful.

\begin{lemma} \label{lm:lemma}
If a linear model has model matrix $F$, and the corresponding unique-treatment model has model matrix $Z$, with $H_Z = Z\left(Z^TZ\right)^{-1} Z$, then $F = H_ZF$.
\end{lemma}

\begin{proof}

Without loss of generality, unique-treatment $\bar{\bx}_j$ is repeated $n_j \ge 1$ times, for $j=1,\dots,q$. Then $n = \sum_{j=1}^q n_j$. Let $\bar{\bff}_j = \bff(\bar{\bx}_j)$, for $j=1,\dots,q$. Then the model matrix $F$ can be written
$$F = \left(
\begin{array}{c}
\bar{F}_1\\
\vdots\\
\bar{F}_q 
\end{array}\right) 
$$
where $\bar{F}_j$ is an $n_j \times p$ matrix with each row given by $\bar{\bff}_j$, i.e. the rows of $\bar{F}_j$ are identical. The corresponding $H_Z$ matrix can be written in the following block diagonal form
$$H_Z = \left( \begin{array}{ccc}
\frac{1}{n_1} J_{n_1} & 0 & \\
0 & \ddots & 0 \\
 & 0 & \frac{1}{n_q}J_{n_q}
 \end{array} \right),
$$
where $J_{n_j}$ is the $n_j \times n_j$ matrix of ones.
Then
$$H_Z F = \left( \begin{array}{c}
\frac{1}{n_1} J_{n_1}\bar{F}_1 \\
\vdots \\
\frac{1}{n_q} J_{n_q}\bar{F}_q
\end{array}\right),$$
where $J_{n_j} \bar{F}_j= n_j \bar{F}_j$. Therefore $H_Z F = F$ as required.
\end{proof}

\begin{lemma} \label{lm:lemma2}
If the designer distribution, $\mathcal{D}(\br,X)$, is such that $\by \sim \mathrm{N}\left( Z \bar{\bmu}, \kappa I_n \right)$, then
\begin{enumerate}
\item[(i)]
$\Vert \by \Vert_{H_F}$ and $\Vert \by \Vert_{I_n - H_Z}$ are independent;
\item[(ii)]
$F^T\by$ and $\Vert \by \Vert_{I_n - H_Z}$ are independent;
\item[(iii)]
$\Vert \by \Vert_{I_n - H_Z} \sim \chi^2_{d}$;
\item[(iv)]
\begin{eqnarray*}
\mathrm{E}_{\mathcal{D}(\br,X)}\left[ \frac{1}{\Vert \by \Vert_{I_n - H_Z}} \right] &=& \frac{1}{\kappa (d-2)}\\
\mathrm{E}_{\mathcal{D}(\br,X)}\left[ \log \Vert \by \Vert_{I_n - H_Z} \right] &=& \log \kappa + \log 2 + \psi \left(\frac{d}{2}\right).
\end{eqnarray*}
\end{enumerate}
\end{lemma}

\begin{proof}

Statements (i) to (iv) are proved as follows.

\begin{enumerate}
\item[(i)]
From Lemma~\ref{lm:lemma}, $\left(I_n - H_Z \right)H_F = H_F - H_ZH_F = H_F - H_Z F (F^TF)^{-1} F^T = H_F - F (F^TF)^{-1} F^T = H_F - H_F = 0_{n \times n}$. Statement (i) follows from Craig's theorem (see, for example, \citealt{hocking_2013}, Theorem 10.2).
\item[(ii)]
From Lemma~\ref{lm:lemma}, $F^T \left(I_n - H_Z \right) = \left[ \left(I_n - H_Z \right) F \right]^T = 0_{n \times n}$. Statement (ii) follows from, for example, \citealt{hocking_2013}, Theorem 10.3.
\item[(iii)]
Since $I_n - H_Z$ is idempotent and
\begin{eqnarray*}
\Vert Z \bar{\bmu} \Vert_{I_n - H_Z} & = & \bar{\bmu}^T Z^T Z \bar{\bmu} - \bar{\bmu}^T Z^T H_Z Z \bar{\bmu}\\
& = & \bar{\bmu}^T Z^T Z \bar{\bmu} - \bar{\bmu}^T Z^T Z \bar{\bmu}\\
& = & 0,
\end{eqnarray*}
statement (iii) follows from, for example, \citealt{hocking_2013}, Theorem 10.1.
\item[(iv)]
The expectations in statement (iv) follow from properties of $\chi^2_d$.
\end{enumerate}
\end{proof}

The target parameter values minimise
\begin{eqnarray*}
L_{SS, \mathcal{D}}(\bt; \br, X) & = & \mathrm{E}_{\mathcal{D}(\br, X)} \left[ \frac{\Vert \by - F \bt \Vert_{I_n}}{2\tilde{\sigma}^2} \right]\\
& = & \frac{n-q}{2} \mathrm{E}_{\mathcal{D}(\br, X)} \left[ \frac{\Vert \by - F \bt \Vert_{I_n}}{\Vert \by \Vert_{I_n- H_Z}} \right]
\end{eqnarray*}
where $\mathcal{D}(\br,X)$ has $\by \sim \mathrm{N}\left( Z \bar{\bmu}, \kappa I_n \right)$ and $\br = \left(\mu(\cdot), \kappa, \sigma^2 \right)$. Differentiating $L_{SS, \mathcal{D}}(\bt; \br, X)$ with respect to $\bt$ and setting equal to $\bzero_p$ gives the following target parameter values
$$\btheta_{SS, \mathcal{D}}(\br, X) = \frac{1}{\mathrm{E}_{\mathcal{D}(\br, X)} \left( \Vert \by \Vert_{I_n - H_Z}^{-1} \right)} \left(F^TF\right)^{-1} \mathrm{E}_{\mathcal{D}(\br, X)} \left( \frac{F^T\by}{\Vert \by \Vert_{I_n - H_Z}} \right).$$
From Lemma~\ref{lm:lemma2} (ii), $F^T\by$ and $\Vert \by \Vert_{I_n - H_Z}$ are independent, so
\begin{eqnarray*}
\btheta_{SS, \mathcal{D}}(\br, X) &=& \frac{1}{\mathrm{E}_{\mathcal{D}(\br, X)} \left( \Vert \by \Vert_{I_n - H_Z}^{-1} \right)} \left(F^TF\right)^{-1} \mathrm{E}_{\mathcal{D}(\br, X)} \left( F^T\by \right) \mathrm{E}_{\mathcal{D}(\br, X)} \left( \Vert \by \Vert_{I_n - H_Z}^{-1} \right),\\
&=& \left(F^TF\right)^{-1}F^T \bmu,
\end{eqnarray*}
as required.

\section{Justification of Gibbs expected utilities given by equation (7)} \label{sec:SM2}

\subsection{NSE utility}

The Gibbs posterior mean is given by
$$\mathrm{E}_{\mathcal{G}}\left[ \btheta_{SS,\mathcal{D}}(\br;X)\right] = \hat{\btheta}_{SS}(\by;X).$$
Therefore, the inner expectation (i.e. with respect to $\by$ under $\mathcal{D}(\br;X)$) in the Gibbs expected NSE utility is 
\begin{eqnarray*}
\mathrm{E}_{\mathcal{D}(\br;X)} \left[ u_{\mathcal{G},NSE} \left( \btheta_{SS,\mathcal{D}}(\br;X), \by, X \right)\right] & = & \mathrm{E}_{\mathcal{D}(\br;X)} \left[ \Vert \btheta_{SS,\mathcal{D}}(\br;X) - \hat{\btheta}_{SS}(\by;X) \Vert_2^2 \right]\\
& = & - \sigma^2 \mathrm{tr}\left[ \left(F^TF \right)^{-1}\right],
\end{eqnarray*}
since $\mathrm{E}_{\mathcal{D}(\br;X)} \left[\hat{\btheta}_{SS}(\by;X)\right] = \btheta_{SS,\mathcal{D}}(\br;X)$. Now the Gibbs expected NSE utility is
\begin{eqnarray*}
U_{\mathcal{G},NSE}(X) & = & \mathrm{E}_{\mathcal{C}}\left[ - \sigma^2 \mathrm{tr}\left[ \left(F^TF \right)^{-1}\right]\right],\\
& = &  \mathrm{E}_{\mathcal{C}}\left( - \sigma^2 \right) \mathrm{tr}\left[ \left(F^TF \right)^{-1}\right].
\end{eqnarray*}
A constraint of $n > q$ ensures that the Gibbs posterior distribution exists.

\subsection{SH utility}

The Gibbs posterior distribution is $\mathrm{N}\left[\hat{\btheta}_{SS}(\by;X), \hat{\sigma}^2 \left(F^TF\right)^{-1}\right]$ with log pdf
\begin{eqnarray}
\log \pi_\mathcal{G}\left(\bt \vert \by; X\right) &=& -\frac{p}{2} \log \left(2\pi\right) - \frac{p}{2}\log \hat{\sigma}^2 + \frac{1}{2} \log \vert F^TF \vert - \frac{1}{2\hat{\sigma}^2} \left[ \bt - \hat{\btheta}_{SS}(\by;X)\right]^T F^TF\left[ \bt - \hat{\btheta}_{SS}(\by;X)\right] \nonumber \\
& = & -\frac{p}{2} \log \left(2\pi\right) + \frac{p}{2}\log d - \frac{p}{2}\log \left[ \by^T \left(I_n - H_Z \right) \by \right] + \frac{1}{2}\log \vert F^TF \vert \nonumber \\
& & - \frac{d \bt^T F^TF \bt}{2\by^T \left(I_n - H_Z \right) \by} +  \frac{d \bt^T F^T\by}{\by^T \left(I_n - H_Z \right) \by} - \frac{d \by^T H_F\by}{2\by^T \left(I_n - H_Z \right) \by}, \label{eqn:sh1}
\end{eqnarray}
where $H_F = F \left(F^TF\right)^{-1}F$. The line (\ref{eqn:sh1}) follows from $\hat{\sigma}^2 = \by^T \left(I_n - H_Z \right) \by/(n-q)$ and $d = n - q$ is the pure error degrees of freedom. 

From Lemma~\ref{lm:lemma2} (i), (ii), and (iv), it follows that the inner expectation (i.e. with respect to $\by$ under $\mathcal{D}(\br;X)$) in the Gibbs expected SH utility is 

\begin{eqnarray}
\mathrm{E}_{\mathcal{D}(\br;X)} \left[ u_{\mathcal{G},SH} \left( \btheta_{SS,\mathcal{D}}(\br;X), \by, X \right)\right] & = & \mathrm{E}_{\mathcal{D}(\br;X)} \left[ \log \pi_\mathcal{G}\left(\btheta_{SS,\mathcal{D}}(\br;X)\vert \by; X\right) \right] \nonumber \\
& = & -\frac{p}{2} \left[ \log (2\pi) + \log \sigma^2 + \log 2\right] - \frac{p}{2} h_2(d) + \frac{1}{2}\log \vert F^TF \vert, \label{eqn:sh2}
\end{eqnarray}
since $\btheta_{SS,\mathcal{D}}(\br;X) = \left(F^TF\right)^{-1}F^T Z\bar{\bmu}$. Equation (7) follows from taking expectation of (\ref{eqn:sh2}) with respect to the hyper-variable distribution $\mathcal{C}$.

\section{Simulation study comparing fixed and random specifications of the calibration weight} \label{sec:SM3}

In this section, we present the results of a simulation study comparing the fixed and random specifications of the calibration weight, under the SS loss.

We let $k=3$ and consider the same regression function from Sections~4.2.3 and~4.3.2, i.e.
$$\bff(\bx) = \left(1, x_1,x_2,x_3,x_1^2,x_2^2,x_3^2,x_1x_2,x_1x_3,x_2x_3\right)^T.$$
We generate responses $\by$ via a uniformly distributed random design of $n$ runs using the designer distribution described in Section~4.2.3. We evaluate the Gibbs posterior distributions under both the fixed and random specifications of the calibration weight. We repeat this $B=1000$ times for $n=20,25,30,\dots,200$. 

The approaches to specify the calibration weights are compared using the median of the mean coverage and length (measuring precision) of the 95\% probability intervals for $\btheta_{SS, \mathcal{D}}(\br,X)$. The mean is over the $B=1000$ repetitions and the median is over the $p=10$ elements of $\btheta_{SS, \mathcal{D}}(\br,X)$.

Figure~\ref{fig:sim_study} shows plots of the coverage and length against $n$. The random calibration weight approach appears to results in wider (less precise) 95\% probability intervals that result in over-coverage. Moreover, the coverage of the 95\% probability intervals does not appear to converge to 95\% as $n$ increases (in contrast to the fixed calibration weight approach). This is due to the random calibration weight approach using the estimate of the error variance which is based on the incorrect regression function.

\begin{figure}
    \includegraphics[scale=0.9]{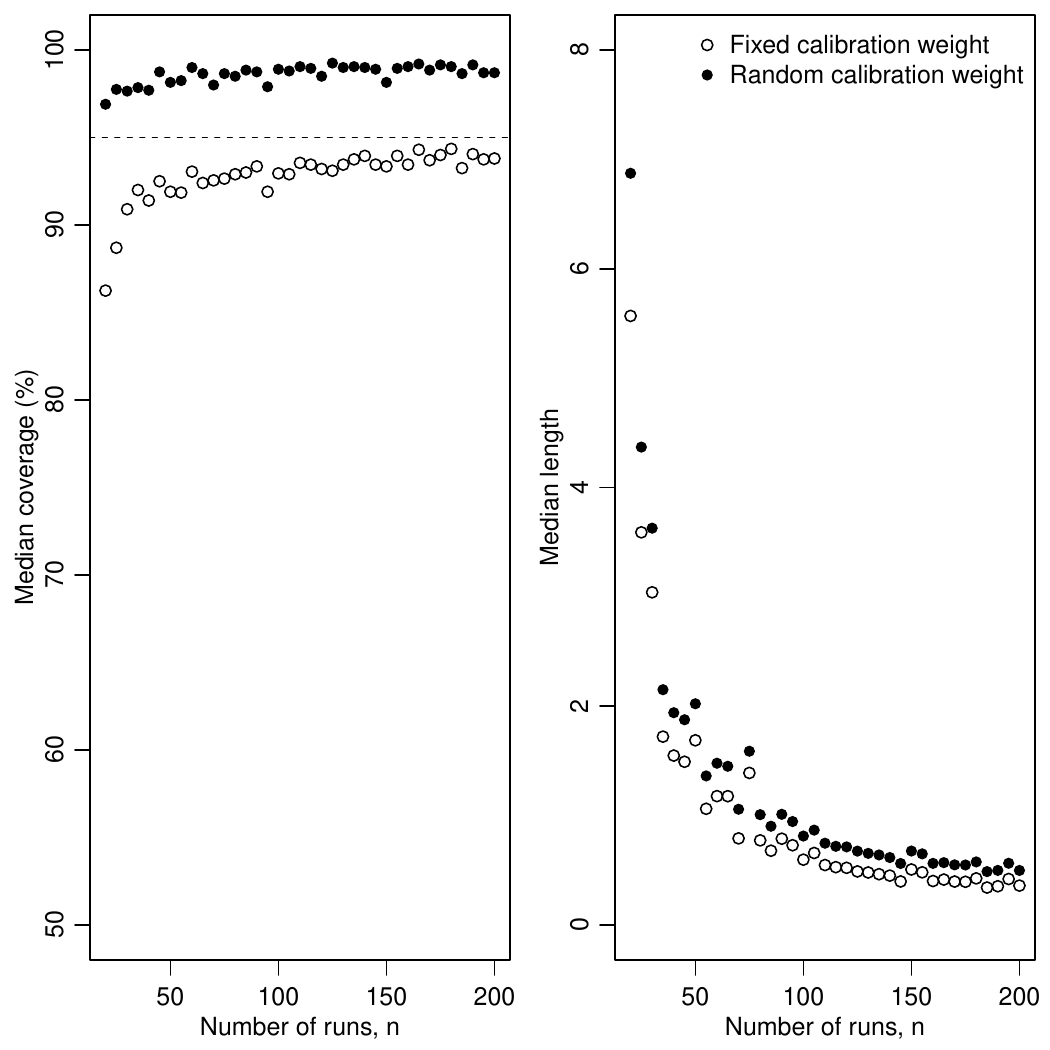}
    \caption{Plots coverage (left) and length (right) against $n$ for the two different approaches to specifying the calibration weight.}
    \label{fig:sim_study}
\end{figure}

\bibliographystyle{rss}

\bibliography{biblio}
\end{document}